\documentclass[12pt,a4paper]{amsart}
\usepackage{amssymb,amsfonts,amsmath,amsthm,bbm}
\usepackage{fullpage}
\usepackage[T1]{fontenc}
\input xy
\xyoption{all}

\newtheorem{thm}{Theorem} 
\newtheorem{prop}[thm]{Proposition}
\newtheorem{lma}[thm]{Lemma}
\newtheorem{cor}[thm]{Corollary}

\newtheorem{ex}{Example}

\newcommand{\CC}{\mathbb{C}}
\newcommand{\KK}{\mathcal{K}}
\newcommand{\U}{\mathcal{U}}
\renewcommand{\u}{\mathfrak{u}}
\newcommand{\I}{\mathbbm{I}}
\renewcommand{\1}{{\mathbf 1}}
\newcommand{\0}{{\mathbf 0}}

\newcommand{\tr}{\operatorname{Tr}}

\newcommand{\Dinv}{\mathcal{D}_\text{inv}}
\newcommand{\Sinv}{\mathcal{S}_\text{inv}}
\newcommand{\diag}{\operatorname{diag}}
\newcommand{\Ker}{\operatorname{Ker}}
\newcommand{\T}{\operatorname{T}\!}
\renewcommand{\H}{\operatorname{H}\!}
\newcommand{\V}{\operatorname{V}\!}

\newcommand{\HH}{\mathcal{H}}
\newcommand{\D}{\mathcal{D}}
\newcommand{\A}{\mathcal{A}}

\newcommand{\half}{\frac 12}

\newcommand{\dt}{\operatorname{d}\!t}

\newcommand{\eps}{\varepsilon}

\newcommand{\ket}[1]{|{#1}\rangle}

\newcommand{\ketbra}[2]{|#1\rangle\langle #2|}
\newcommand{\dd}[1]{\frac{\operatorname{d}}{\operatorname{d}\!#1}}

\newcommand{\dist}{\operatorname{dist}}

\title{Dynamic Distance Measures on Spaces of Isospectral Mixed Quantum States}
\author{Ole Andersson \and  Hoshang Heydari}

\address{Department of Physics, Stockholm University, 10691 Stockholm, Sweden}


\begin{document}
\begin{abstract}
Distance measures are indispensable tools in quantum information processing and quantum computing. This since they can be used to quantify to what extent information is preserved, or altered, by quantum processes. In this paper we propose a new distance measure for mixed quantum states, that we call the dynamic distance measure, and show that it is a proper distance measure. The dynamic distance measure is defined in terms of a measurable quantity, which make it very suitable for applications. In a final section we compare the dynamical distance measure with the well-known Bures distance.
\end{abstract}

\maketitle
\section{Introduction}

Quantum information has the reputation of being a futuristic field full of far-reaching promises.
The field has attracted researchers from many different branches of science and engineering whose 
efforts have greatly improved our understanding of the physical nature of information, and hopefully will provide us with new cutting-edge technological innovations in the future. Quantum information theory has been applied to such diverse areas as bio science, nano-technology, economics, and game theory \cite{Ohya_etal2011, Sharif_etal2012, Landsburg2011}.

In recent years, new experimental results has shed light on some murky and hidden parts of quantum information, and has also opened up new opportunities beyond our expectations. Furthermore, new theoretical tools, mostly from geometry and topology, has been successfully applied to the field. For example, geometrical considerations led to the important characterization of entanglement, and the development of efficient, error-prone systems for quantum computers.

Distance measures are some of the most basic geometrical tools used in quantum information theory. Such measures are, for example, used to compare the input and output of quantum channels and gates, and hence to quantify to what extent information is preserved, or altered, by quantum processes. Examples of well-known distance measures are the trace-distance, fidelity, and Bures distance \cite{Nielsen_etal2010,Uhlmann1986, Uhlmann1991}. In this paper, we propose a new distance measure that we call the dynamic distance measure. This distance measure is defined for all pairs of isospectral, i.e., unitarily equivalent, mixed states.
We show that the dynamic distance is a proper distance measure -- a verification that, despite the naturalness of the definition, requires a surprisingly extensive geometric machinery.

Here is the outline of the paper. In Section \ref{IMSDDM} we define the dynamic distance measure, and state the main result. In Section \ref{SPIMS} we introduce a geometric framework for mixed quantum states. In Section \ref{PDDM} we, in detail, discuss properties of the dynamic distance measure and prove the main result. Finally, in Section \ref{Uhlmann and Bures} we compare the dynamic distance measure and the Bures distance \cite{Bures1969,Uhlmann1992}.

\section{Isospectral mixed states and dynamic distance measures}\label{IMSDDM}

Mixed quantum states can be represented by density operators, i.e., self-adjoint, nonnegative, trace-class operators with unit trace.
We denote the space of density operators for a quantum system modeled on a Hilbert space $\HH$ by $\D(\HH)$, and its subspace of density operators with finite rank at most $k$ by $\D_k(\HH)$.

A density operator that evolves according to a von Neumann equation remains in a single orbit of the left conjugation action of the unitary group of $\HH$ on $\D(\HH)$. The orbits are in one-to-one correspondence with the possible spectra for density operators on $\HH$. By a spectrum of a density operators of rank $k$ we mean its non-increasing sequence $\sigma=(p_1,p_2,\dots,p_k)$ of positive eigenvalues, repeated in accordance with their multiplicity. We henceforth assume $\sigma$ to be fixed, and write $\D(\sigma)$ for the corresponding orbit.

Suppose $\rho_0$ and $\rho_1$ are two density operators in $\D(\sigma)$. 
Let $H$ be a Hamiltonian operator on $\HH$, 
and assume that a curve $\rho$ satisfies the boundary value von Neumann equation:
\begin{equation}\label{von}
i\dot\rho=[H,\rho],\qquad \rho(t_0)=\rho_0,\quad\rho(t_1)=\rho_1.
\end{equation}
We then define the \emph{$H$-distance} from $\rho_0$ and $\rho_1$ to be the path integral of the uncertainty of $H$ along $\rho$,
\begin{equation*}\label{Hdistance}
D_H(\rho_0,\rho_1)=\int_{t_0}^{t_1}\sqrt{\tr(H^2\rho)-\tr(H\rho)^2}\dt.
\end{equation*}
This distance measure has an interesting physical interpretation -- it can be regarded the energy input needed to run $\rho_0$ into $\rho_1$ using $H$.
We also define the \emph{dynamic distance} between $\rho_0$ and $\rho_1$ to be
\begin{equation*}
D(\rho_0,\rho_1)=\inf_H D_H(\rho_0,\rho_1),
\end{equation*}
where the infimum is taken over all Hamiltonians $H$ for which the boundary value problem \eqref{von}
has a solution. The dynamic distance measure is defined for each pair of isospectral density operators because any two such can be connected by a solution to some von Neumann equation.
The main result of this paper is that the dynamic distance measure  is a proper distance measure on $\D(\sigma)$.

\begin{thm}\label{main result}
The dynamic distance measure is a proper distance measure.
\end{thm}
\noindent Recall that a distance function must satisfy the following conditions:\vspace{3pt}

\begin{tabular}[tb]{rl}
Positivity: & $\dist(\rho_0,\rho_1)\geq 0$.\\
Non-degeneracy: & $\dist(\rho_0,\rho_1)=0\iff \rho_0=\rho_1$.\\
Symmetry: & $\dist(\rho_0,\rho_1)=\dist(\rho_1,\rho_0)$.\\
Triangle inequality: & $\dist(\rho_0,\rho_2)\leq \dist(\rho_0,\rho_1)+\dist(\rho_1,\rho_2)$.
\end{tabular}\vspace{3pt}

\noindent One can show that the dynamic distance measure  also satisfies the following unitary invariance:
\begin{equation*}
D(U\rho_0 U^\dagger,U\rho_1U^\dagger)=D(\rho_0,\rho_1).
\end{equation*}
The proof of Theorem \ref{main result} will be based on a fairly involved mathematical setup.

\section{Standard purification of isospectral mixed states}\label{SPIMS}

A state is called pure if its density operator has rank $1$.
In quantum mechanics, especially quantum information theory, \emph{purification} refers to the fact that every density operator can be thought of as representing the reduced state of a pure state.
More precisely, if $\rho$ is a density operator acting on $\HH$, and $\KK$ is a Hilbert space of large enough dimension, then there is a normalized ket $\ket{\Psi}$ in $\HH\otimes \KK$ such that $\rho$ is the partial trace of $\ketbra{\Psi}{\Psi}$ with respect to $\KK$.
By the \emph{standard purification} of density operators on $\HH$ of rank at most $k$ we will mean the surjective map
$\pi:\mathcal{S}(\HH\otimes\CC^{k*})\to\D_k(\HH)$
defined by $\pi\ket{\Psi}=\tr_{\CC^{k*}}\ketbra{\Psi}{\Psi}$.
Here, $\CC^{k*}$ is the space of linear functionals on $\CC^k$ and $\mathcal{S}(\HH\otimes \CC^{k*})$ is the unit sphere in $\HH\otimes \CC^{k*}$.
If we canonically identify $\HH\otimes \CC^{k*}$ with the space $\mathcal{L}(\CC^k,\HH)$ of linear maps from $\CC^k$ to $\HH$, equipped with the Hilbert-Schmidt inner product, then
\begin{equation}\label{standard purification}
\pi(\Psi)=\Psi\Psi^\dagger.
\end{equation}

Write $P(\sigma)$ for the diagonal $k\times k$ matrix that has $\sigma$ as its diagonal, and let $\mathcal{S}(\sigma)$ be the space of those $\Psi$ in $\mathcal{L}(\CC^k,\HH)$ that satisfies $\Psi^\dagger \Psi=P(\sigma)$, when
$\Psi^\dagger\Psi$ is expressed as a matrix relative to the standard basis in $\CC^k$. Then $\mathcal{S}(\sigma)$ is a subspace of 
the unit sphere in $\mathcal{L}(\CC^k,\HH)$, and the standard purification \eqref{standard purification} restricted to $S(\sigma)$  
is a principal fiber bundle over $\D(\sigma)$ with right acting gauge group $\U(\sigma)$,
\begin{equation}\label{action}
U\cdot\psi=\psi U,\qquad U\in\U(\sigma),\quad \psi\in \mathcal{S}(\sigma),
\end{equation}
consisting of those unitaries in $\U(k)$ that commutes with $P(\sigma)$.
The following two special cases are well known.
\begin{ex}
If $\sigma=(1;1)$, then $\D(\sigma)$ is the complex projective $n$-space,
$\mathcal{S}(\sigma)$ is the $(2n+1)$-dimensional unit sphere, and $\pi$ is the generalized Hopf bundle.
\end{ex}
\begin{ex}
 If $\sigma=(1/k;k)$, then $\D(\sigma)$ is the Grassmann manifold of $k$-planes in $\HH$,
$\mathcal{S}(\sigma)$ is the Stiefel manifold of $k$-frames in $\HH$, and $\pi$ is the Stiefel bundle. 
\end{ex}

The vertical and horizontal bundles over $\mathcal{S}(\sigma)$ are the subbundles
$\V\mathcal{S}(\sigma)=\Ker d \pi$ and $\H\mathcal{S}(\sigma)=\V\mathcal{S}(\sigma)^\bot$
of the tangent bundle of $\mathcal{S}(\sigma)$. Here  $^\bot$ denotes orthogonal complement with respect to the Hilbert-Schmidt product.
Vectors in $\V\mathcal{S}(\sigma)$ and $\H\mathcal{S}(\sigma)$
are called vertical and horizontal, respectively,
and a curve in $\mathcal{S}(\sigma)$ is called horizontal if its velocity vectors are horizontal.
Recall that for every curve $\rho$ in $\D(\sigma)$ and every $\Psi_0$ in the fiber over the initial operator $\rho(t_0)$, there is a unique horizontal lift of $\rho$ to $\mathcal{S}(\sigma)$ that extends from $\Psi_0$ \cite[page 69, Prop 3.1]{Kobayashi_etal1996}.
For convenience, we tacitly assume that all curves in this paper are defined on a common unspecified interval $t_0\leq t\leq t_1$. Moreover, we assume that they are piecewise smooth.

The infinitesimal generators of the gauge group action \eqref{action} yield canonical isomorphisms between the Lie algebra $\u(\sigma)$  of $\U(\sigma)$ and the fibers in $\V\mathcal{S}(\sigma)$. The Lie algebra consists of all anti-Hermitian $k\times k$ matrices that commutes with $P(\sigma)$, and the isomorphisms are
\begin{equation}\label{eq:inf gen}
\u(\sigma)\ni\xi\mapsto \Psi\xi\in\V_\Psi\mathcal{S}(\sigma).
\end{equation}
Furthermore, $\H\mathcal{S}(\sigma)$ is the kernel bundle of the gauge invariant mechanical connection form
$\A_{\Psi}=\I_{\Psi}^{-1}J_{\Psi}$,
where $\I_{\Psi}:\u(\sigma)\to \u(\sigma)^*$ and $J_{\Psi}:\T_{\Psi}{\mathcal{S}(\sigma)}\to \u(\sigma)^*$ are the locked inertia tensor and moment map, respectively,
\begin{equation}\label{eq:beta}
\I_{\Psi}\xi\cdot \eta=\half\tr\left(\left(\xi^\dagger \eta+\eta^\dagger \xi\right)P(\sigma)\right),\qquad
J_{\Psi}(X)\cdot\xi=\half\tr\big(X^\dagger\Psi\xi+\xi^\dagger\Psi^\dagger X\big).
\end{equation}
Using \eqref{eq:beta} we can derive an explicit formula for the connection form.
If $m_1, m_2, \dots , m_l$ are the multiplicities of the different eigenvalues in $\sigma$, with $m_1$ being the multiplicity of the greatest eigenvalue, $m_2$ the multiplicity of the second greatest eigenvalue, etc., and if for $j=1,2,\dots,l$,
\begin{equation*}
E_j=\diag(\0_{m_1},\dots,\0_{m_{j-1}},\1_{m_j},\0_{m_{j+1}},\dots,\0_{m_l}),
\end{equation*}
then
\begin{equation*}
\begin{split}
\I_\Psi\Big(\sum_jE_j\Psi^\dagger XE_jP(\sigma)^{-1}\Big)\cdot\xi
&=\half\tr\Big(\Big(\sum_jP(\sigma)^{-1}E_jX^\dagger\Psi E_j\xi+\xi^\dagger\sum_j E_j\Psi^\dagger XE_jP(\sigma)^{-1}\Big)P(\sigma)\Big)\\
&=\half\tr\Big(\sum_jE_j(X^\dagger\Psi \xi+\xi^\dagger\Psi^\dagger X)E_j\Big)\\
&=\half\tr\big(X^\dagger\Psi\xi+\xi^\dagger\Psi^\dagger X\big)\\
&=J_\Psi(X)\cdot\xi
\end{split}
\end{equation*}
for every $X$ in $\T_\Psi\mathcal{S}(\sigma)$ and every $\xi$ in $\u(\sigma)$.
Hence
\begin{equation*}\label{eq:explicit}
\A_\Psi(X)=\sum_jE_j\Psi^\dagger XE_jP(\sigma)^{-1}.
\end{equation*}
Observe that the orthogonal projection of $\T_\Psi\mathcal{S}(\sigma)$ onto $\V_\Psi\mathcal{S}(\sigma)$
is given by the connection form followed by the identification \eqref{eq:inf gen}. Therefore, the vertical and horizontal projections of $X$ in $\T_\Psi\mathcal{S}(\sigma)$ are $X^\bot=\Psi\A_\Psi(X)$ and $X^{||}=X-\Psi\A_\Psi(X)$, respectively.
We finish this section with a discussion on the distance between the fibers of $\pi$ over two given density operators $\rho_0$ and $\rho_1$ in $\D(\sigma)$. 

Consider the space $\Omega(\rho_0,\rho_1)$ of piecewise smooth curves that start in the fiber $\pi^{-1}(\rho_0)$ and end in the fiber $\pi^{-1}(\rho_1)$. This space can be given a natural smooth structure such that the tangent space at a curve $\Psi$ consists of all smooth vector fields $\chi$ along $\Psi$ that are vertical at the end points of $\Psi$. Let $E$ be the energy functional on $\Omega(\rho_0,\rho_1)$,
\begin{equation*} 
E[\Psi]=\half\int_{t_0}^{t_1} \tr(\dot \Psi^\dagger\dot \Psi)\dt.
\end{equation*}
The differential of $E$ at $\Psi$ is given by
\begin{equation*}
dE[\Psi]\chi=\frac{1}{2}\left[\tr(\chi^\dagger\dot \Psi+\dot\Psi^\dagger\chi)\right]_{t_0}^{t_1}-\frac{1}{2}\int_{t_0}^{t_1} \tr(\chi^\dagger \nabla_t\dot \Psi+\nabla_t\dot{\Psi}^\dagger\chi)\dt,
\end{equation*}
where $\nabla_t\dot{\Psi}$ denotes the covariant derivative of $\dot{\Psi}$ along $\Psi$. We call $\Psi$ an \emph{extremal for $E$} if $dE[\Psi]=0$. Clearly, extremals for $E$ are geodesics: $\nabla_t\dot \Psi=0$.

The length of a curve $\Psi$ in $\mathcal{S}(\sigma)$ is
\begin{equation*}
L[\Psi]=\int_{t_0}^{t_1}\sqrt{\tr(\dot{\Psi}^\dagger\dot{\Psi})}\dt.
\end{equation*}
Moreover, the distance between $\pi^{-1}(\rho_0)$ and $\pi^{-1}(\rho_1)$ is defined as the infimum of the lengths of all 
curves in $\Omega(\rho_0,\rho_1)$.
There is at least one curve in $\Omega(\rho_0,\rho_1)$ whose length equals the distance between the two fibers. This since the fibers are compact. Also, every such curve is an extremal for $E$. Therefore, they are horizontal:
\begin{prop}\label{Noethers theorem for geodesics}
If $\Psi$ is a geodesic in $\mathcal{S}(\sigma)$, then $J_\Psi(\dot \Psi)$ is constant.
\end{prop}

\begin{proof} 
Choose any $\eta$ in $\u(\sigma)$ and consider the variation $\Psi_\eps(t)=\Psi(t)\exp(\eps\eta)$. We have that
$\tr(\dot \Psi_\eps^\dagger\dot \Psi_\eps)=\tr(\dot \Psi^\dagger\dot \Psi)$
since $\U(\sigma)$ acts through isometries. Hence
\begin{equation*}
\begin{split}
0&=\frac{1}{2}\dd{\eps}\left[\int_{\tau_0}^{\tau_1}\tr(\dot\Psi_\eps^\dagger\dot\Psi_\eps)\dt\right]_{\eps=0}\\
&=\frac{1}{2}\big[\tr(\eta^\dagger\Psi^\dagger \dot\Psi+\dot{\Psi}^\dagger\Psi\eta)\big]_{\tau_0}^{\tau_1}
-\frac{1}{2}\int_{\tau_0}^{\tau_1} \tr(\eta^\dagger\Psi^\dagger \nabla_t\dot \Psi + \nabla_t\dot \Psi^\dagger \Psi\eta)\dt\\
&=\big[J_\Psi(\dot \Psi)\cdot\eta\big]_{\tau_0}^{\tau_1}
\end{split}
\end{equation*}
for any $t_0\leq \tau_0\leq \tau_1\leq t_1$. We conclude that $J_\Psi (\dot \Psi)$ is constant.
\end{proof}

\section{Properties of the dynamic distance measure}\label{PDDM}

I this section we will prove that 
$D(\rho_0,\rho_1)$ equals the distance between the fibers $\pi^{-1}(\rho_0)$ and $\pi^{-1}(\rho_1)$. Theorem \ref{main result} follows easily from this observation.

\begin{prop}\label{crucial}
Suppose $\rho$ solves \eqref{von}. Let $\Phi$ be a horizontal lift of $\rho$. Then $D_H(\rho_0,\rho_1)\geq L[\Phi]$. Moreover, $D_H(\rho_0,\rho_1)=L[\Phi]$ if $i\dot\Phi=H\Phi$.
\end{prop}
\begin{lma}\label{Jensen}
We have that $\tr(\xi^2 P(\sigma))\leq\tr(\xi P(\sigma))^2$
for every $\xi$ in $\u(\sigma)$, and $\tr(\xi^2 P(\sigma))=\tr(\xi P(\sigma))^2$ if and only if $\xi$ is a constant multiple of the identity.
\end{lma}

\begin{proof}
Write $i\xi=U\delta U^\dagger$, where $\delta$ is a real diagonal matrix and $U$ belongs to $\U(\sigma)$. We have that $\tr(\xi^2P(\sigma))=-\tr(\delta^2P(\sigma))$ and $\tr(\xi P(\sigma))^2=-\tr(\delta P(\sigma))^2$ since $P(\sigma)$ commutes with $U$. Moreover, $\tr(\delta P(\sigma))^2\leq \tr(\delta^2P(\sigma))$, and $\tr(\delta P(\sigma))^2= \tr(\delta^2P(\sigma))$ if and only if $\delta$ is a constant multiple of the identity. This since $x\mapsto x^2$ is strongly convex.
\end{proof}

\begin{proof}[Proof of Proposition \ref{crucial}]
Let $\Phi$ be a horizontal lift of $\rho$, and let $\Psi$ be any lift of $\rho$ such that $i\dot \Psi=H\Psi$ and $\Psi(t_0)=\Phi(t_0)$.
Then
\begin{equation*}
\Phi=\Psi U,\qquad U(t)=\exp_{+}\left(-\int_{t_0}^{t}\A_\Psi(\dot \Psi)\dt\right),
\end{equation*}
where $\exp_{+}$ is the positive time-ordered exponential. Now,
\begin{equation}\label{for any}
\begin{split}
\tr(H^2\rho)-\tr(H\rho)^2
&=\tr(\Psi^\dagger H^2 \Psi)-\tr(\Psi^\dagger H\Psi)^2\\
&=\tr(\dot \Psi^\dagger\dot \Psi)+\tr(\Psi^\dagger\dot \Psi)^2\\
&=\tr(\dot \Psi^\dagger\dot \Psi)+\tr(E_j\Psi^\dagger\dot \Psi E_j)^2\\
&=\tr(\dot \Psi^\dagger\dot \Psi)+\tr(\A_\Psi(\dot \Psi)P(\sigma))^2,
\end{split}
\end{equation}
and
\begin{equation*} 
\begin{split}
\tr(\dot \Phi^\dagger\dot \Phi)
&=\tr\left(U^\dagger\left(\dot \Psi^\dagger+\A_\Psi(\dot \Psi)\Psi^\dagger)(\dot \Psi-\Psi\A_\Psi(\dot \Psi)\right)U\right)\\
&=\tr\left(\dot \Psi^\dagger\dot \Psi+(\Psi^\dagger\dot \Psi-\dot \Psi^\dagger \Psi)\A_\Psi(\dot \Psi)-\A_\Psi(\dot \Psi)^2 P(\sigma)\right)\\
&=\tr(\dot \Psi^\dagger\dot \Psi)+2\tr(\Psi^\dagger\dot \Psi\A_\Psi(\dot \Psi))-\tr(\A_\Psi(\dot \Psi)^2 P(\sigma))\\
&=\tr(\dot \Psi^\dagger\dot \Psi)+2\tr(\Psi^\dagger\dot \Psi E_j\Psi^\dagger\dot \Psi E_jP(\sigma)^{-1})-\tr(\A_\Psi(\dot \Psi)^2 P(\sigma))\\
&=\tr(\dot \Psi^\dagger\dot \Psi)+2\tr\left(\left(E_j\Psi^\dagger\dot \Psi E_j\right)^2P(\sigma)^{-1}\right)-\tr(\A_\Psi(\dot \Psi)^2 P(\sigma))\\
&=\tr(\dot \Psi^\dagger\dot \Psi)+\tr(\A_\Psi(\dot \Psi)^2 P(\sigma)).
\end{split}
\end{equation*}
Hence
\begin{equation*} 
\begin{split}
\tr(H^2\rho)-\tr(H\rho)^2
&=\tr(\dot \Psi^\dagger\dot \Psi)+\tr(\A_\Psi(\dot \Psi)P(\sigma))^2\\
&\geq\tr(\dot \Psi^\dagger\dot \Psi)+\tr(\A_\Psi(\dot \Psi)^2 P(\sigma))\\
&=\tr(\dot \Phi^\dagger\dot \Phi)
\end{split}
\end{equation*}
by Lemma \ref{Jensen}. We conclude that
$D_H(\rho_0,\rho_1)\geq L[\Phi]$.
Moreover, if $\Psi$ is horizontal, and thus $\Psi=\Phi$, then $\tr(H^2\rho)-\tr(H\rho)^2=\tr(\dot\Phi^\dagger\dot\Phi)$ according to \eqref{for any}. In this case, $D_H(\rho_0,\rho_1)=L[\Phi]$.
\end{proof}

\begin{cor}\label{distance}
The dynamic distance between two density operators $\rho_0$ and $\rho_1$ in $\D(\sigma)$ equals the distance between the fibers $\pi^{-1}(\rho_0)$ and $\pi^{-1}(\rho_1)$.
\end{cor}

\begin{proof}
Immediate from Proposition \ref{crucial} and the fact that every curve in $\mathcal{S}(\sigma)$ is the solution to some Schr\"odinger equation. This since the unitary group of $\HH$ acts transitively on $\mathcal{S}(\sigma)$. 
\end{proof}

Proposition \ref{crucial} and Corollary \ref{distance} let us conclude that
for any pair of density operators $\rho_0$ and $\rho_1$ in $\D(\sigma)$, there is a curve $\Psi$ in 
$\mathcal{S}(\sigma)$ that extends from the fiber over $\rho_0$ and ends in the fiber over $\rho_1$, and which is such that
$D(\rho_0,\rho_1)=L[\Psi]$. This observation make the proof of Theorem \ref{main result} fairly straightforward.

\begin{proof}[Proof of Theorem \ref{main result}]
The function $D$ is positive because $D_H(\rho_0,\rho_1)$ is always a non-negative number by \eqref{Hdistance}.
Moreover, $D$ is non-degenerate. Indeed, let $\Psi$ be a curve in $\Omega(\rho_0,\rho_1)$ such that 
$D(\rho_0,\rho_1)=L[\Psi]$. If $D(\rho_0,\rho_1)=0$, then $\Psi$ is stationary, and hence $\rho_0=\rho_1$.
The opposite implication is obvious.
 
To see that $D$ is symmetric let $\Psi$ be a curve like the one in the proof of non-degeneracy.
Define $\Phi$ by $\Phi(t)=\Psi(t_1+t_0-t)$. Then $\Phi$ is a horizontal curve that projects onto a curve in $\D(\sigma)$ from $\rho_1$ to $\rho_0$.
Consequently,
\begin{equation*}
D(\rho_1,\rho_0)\leq L[\Phi]=L[\Psi]=D(\rho_0,\rho_1).
\end{equation*}  
An identical argument shows that $D(\rho_0,\rho_1) \leq D(\rho_1,\rho_0)$.
Thus, $D$ is symmetric.

Finally, to see that $D$ satisfies the  triangle inequality  let $\Psi_{ij}$ be a horizontal curve in $\mathcal{S}(\sigma)$ covering a curve in $\D(\sigma)$ from $\rho_i$ to $\rho_j$, $i,j=0,1,2$. Also assume that $D(\rho_i,\rho_j)=L[\Psi_{ij}]$. Then 
\begin{equation*}
\Phi(t)=\begin{cases}
\Psi_{01}(2t-t_0),\quad\text{if}\quad t_0\leq t\leq (t_1+t_0)/2,\\
\Psi_{12}(2t-t_1),\quad\text{if}\quad (t_1+t_0)/2\leq t\leq t_1.\\
\end{cases}
\end{equation*}
is a horizontal curve connecting the fibers over $\rho_0$ and $\rho_2$.
Therefore, 
\begin{equation*}
D(\rho_0,\rho_2)\leq L[\Phi]= L[\Psi_{01}]+L[\Psi_{12}]=D(\rho_0,\rho_1)+D(\rho_1,\rho_2).
\end{equation*}
Hence $D$ satisfies the triangle inequality.
\end{proof}

\section{Relation between the dynamic distance measure  and Bures distance}\label{Uhlmann and Bures}

Suppose $\HH$ is $n$-dimensional. Let $\Sinv(\CC^n,\HH)$ be the space of all
invertible maps in $\mathcal{L}(\CC^n,\HH)$ with unit norm, and $\Dinv(\HH)$  be the space of all invertible density operators acting on $\HH$. Then $\Pi:\Sinv(\CC^n,\HH)\to\Dinv(\HH)$ defined by $\Pi(\Psi)=\Psi\Psi^\dagger$ is a $\U(n)$-bundle, which we call \emph{Uhlmann's bundle} since it first appeared in \cite{Uhlmann1991}.
The geometry of Uhlmann's bundle has been thoroughly investigated, and it is an important tool in quantum information theory, mainly due to its close relationship with the Bures distance function \cite{Bures1969,Uhlmann1992}.

Uhlmann's bundle is equipped with the mechanical connection, which means that the horizontal bundle is the orthogonal complement of the vertical bundle with respect to the Hilbert-Schmidt inner product. Moreover, the Bures distance between two density operators in $\Dinv(\HH)$ equals the distance between the corresponding fibers of $\Pi$, see \cite{Uhlmann1992}. We denote the Bures distance function by $D_B$.

Suppose $\sigma$ has length $n$. Then $\mathcal{S}(\sigma)$ is a submanifold of $\Sinv(\CC^n,\HH)$. Moreover, the vertical bundle of $\mathcal{S}(\sigma)$ is subbundle of the restriction of the vertical bundle of $\Sinv(\CC^n,\HH)$ to $\mathcal{S}(\sigma)$.
However, no nonzero horizontal vector in Uhlmann's bundle is tangential to $\mathcal{S}(\sigma)$.
To see this, let $\Psi$ be any element in $\mathcal{S}(\sigma)$. Then $X$ in $\T_\Psi\Sinv(\CC^n,\HH)$ is horizontal, i.e. is annihilated by the mechanical connection of the Uhlmann bundle, if and only if 
\begin{equation}\label{Uhlmann parallel}
\Psi^\dagger X-X^\dagger\Psi=0,
\end{equation}
see \cite{Uhlmann1991}. On the other hand, every $X$ in $\T_\Psi\mathcal{S}(\sigma)$ satisfies
\begin{equation}\label{relation}
\Psi^\dagger X+X^\dagger\Psi=0
\end{equation}
since $\Psi^\dagger\Psi=P(\sigma)$. Clearly, only the zero vector satisfies both \eqref{Uhlmann parallel} and \eqref{relation}.

The distance between $\rho_0$ and $\rho_1$ in $\D(\sigma)$ is never smaller than Bures distance between them. Indeed, every curve between $\pi^{-1}(\rho_0)$ and $\pi^{-1}(\rho_1)$ in $\mathcal{S}(\sigma)$ is a curve between $\Pi^{-1}(\rho_0)$ and $\Pi^{-1}(\rho_1)$ in $\Sinv(\CC^n,\HH)$, and since the metrics on the total spaces of the two bundles are induced from a common ambient metric we can conclude that
\begin{equation}\label{bures inequality}
D(\rho_0,\rho_1)\geq D_B(\rho_0,\rho_1).
\end{equation}

Uhlmann \cite{Uhlmann1992} and Dittmann \cite{Dittmann1993, Dittmann1999} have derived explicit formulas for the Bures distance for density operators on finite dimensional Hilbert spaces.
For density operators on $\CC^2$ the formula reads
\begin{equation}\label{Dittmann}
D_B(\rho,\rho+\delta\rho)^2=\frac 14\tr\big(\delta\rho\delta\rho+\frac {1}{\det\rho}(\delta\rho-\rho\delta\rho)^2\big).
\end{equation}
We use this formula to show that there are density operators $\rho_0$ and $\rho_1$ acting on $\CC^2$ for which
the inequality in \eqref{bures inequality}
is strict.

Suppose $\sigma=(p_1,p_2)$, let $\eps>0$, and define a curve $\Psi$ in $\mathcal{S}(\sigma)$ by 
\begin{equation*}
\Psi(t)=\begin{bmatrix} \sqrt{p_1}\cos(\eps t) & \sqrt{p_2}\sin(\eps t)\\ -\sqrt{p_1}\sin(\eps t) & \sqrt{p_2}\cos(\eps t)\end{bmatrix},\qquad 0\leq t\leq 1.
\end{equation*}
Set $\rho_0=\Psi(0)\Psi(0)^\dagger$ and $\rho_1=\Psi(1)\Psi(1)^\dagger$. Then, for $\eps$ small enough, the length of $\Psi$ equals $D(\rho_0, \rho_1)$. In this case,
$D(\rho_0,\rho_1)=L[\Psi]=\eps$.
However, \eqref{Dittmann} yields
\begin{equation*}
D_B(\rho_0,\rho_1)=\frac {p_1-p_2}{\sqrt{2}}|\sin\eps|\sqrt{2+\frac{(p_1-p_2)^2}{2p_1p_2}\sin^2\eps}.
\end{equation*}

\section{Conclusion}

In summary we have introduced a measurable quantity, called a dynamic distance measure, on each space of isospectral density operators, and shown that it is a proper distance measure, i.e., a positive, non-degenerate, symmetric binary function that satisfies the triangle inequality.  
The main result was formulated in Section \ref{IMSDDM}, but its proof was postponed until Section \ref{PDDM} 
to make the paper accessible also to those readers who are mainly interested in the result, rather than the extensive geometrical setup and fairly technical proof. We have also compared our dynamic distance measure with the Bures distance. The outcome of that comparison is that the dynamic distance measure and the Bures distance are different. In fact, the dynamic distance measure is bounded from below by the Bures distance.
Because the dynamic distance measure is defined in terms of Hamiltonians, we believe that our results have many interesting applications in fields such as quantum computing and condense matter, where Hamiltonians for specific quantum operations or specific quantum systems are usually defined explicitly.

\section*{Acknowledgments}

The second author acknowledges the financial support from the Swedish Research Council (VR).

\end{document}